\documentclass[12pt]{article}

\usepackage[centertags]{amsmath}
\usepackage{amsthm,amsfonts}

\usepackage{amssymb, latexsym}

 \usepackage{color}
\usepackage[normalem]{ulem}
\usepackage[english]{babel}
 \usepackage{hyperref}
\makeatletter
\def\@seccntformat#1{\csname the#1\endcsname.\ } 
\makeatother

\theoremstyle{plain}
\newtheorem{theorem}{Theorem}[section]
\newtheorem{lemma}[theorem]{Lemma}
\newtheorem{proposition}[theorem]{Proposition}
\newtheorem{corollary}[theorem]{Corollary}
\theoremstyle{definition}

\newtheorem{remark}[theorem]{Remark}
\newtheorem{definition}[theorem]{Definition}
\newtheorem{construction}[theorem]{Construction}

\newcommand{\Ima}{\operatorname{Im}}

\newcommand{\rank}{\operatorname{rank}}
\newcommand{\Norm}{\operatorname{Norm}}

\newcommand{\Tr}{\operatorname{Trace}}

\newcommand{\F}{{\mathbb F}}
\newcommand{\Z}{{\mathbb Z}}

\def\Aut{\mathrm{Aut}}
\def\GL{\mathrm{GL}}
\def\FF{\mathbb{F}}
\def\GF#1{\FF_{#1}}

\newcommand{\be}{\begin{eqnarray}}
\newcommand{\ee}{\end{eqnarray}}
\newcommand{\nn}{{\nonumber}}

\newcommand{\Keywords}[1]{\par\noindent
{\small{\it Keywords\/}: #1}}

\title{Constructing MRD codes by switching%
\thanks{
  The work of M.\,J.\,Shi is supported by the National Natural Science
Foundation of China (12071001);
  the work of
D.\,S.\,Krotov is supported with
the state contract of the
Sobolev Institute of Mathematics (FWNF-2022-0017).
} 
} 

\author{Minjia Shi%
\thanks{Minjia Shi is with Key Laboratory of Intelligent Computing Signal Processing, Ministry of Education,
School of Mathematical Sciences, Anhui University, Hefei, Anhui, 230601, China, e-mail: smjwcl.good@163.com}%
,
 Denis S. Krotov%
 \thanks{Denis Krotov is with  Sobolev Institute of Mathematics, Novosibirsk 630090,
Russia. E-mail: krotov@math.nsc.ru}%
,
 Ferruh \"{O}zbudak%
\thanks{Ferruh \"{O}zbudak is with Department of Mathematics and Institute of Applied Mathematics, Middle East Technical
        University,   Ankara, Turkey;
        e-mail: ozbudak@metu.edu.tr}%
 } 
\date{}

\begin{document}\maketitle

\abstract
MRD codes are maximum codes in the rank-distance metric space on $m$-by-$n$ 
matrices over the finite field of order $q$. They are diameter perfect and have 
the cardinality $q^{m(n-d+1)}$ if $m\ge n$. We define switching in MRD codes as 
replacing special MRD subcodes by other subcodes with the same parameters. We 
consider constructions of MRD codes admitting such switching, including 
punctured twisted Gabidulin codes and direct-product codes. Using switching, we 
construct a huge class of MRD codes whose cardinality grows doubly exponentially 
in $m$ if the other parameters ($n$, $q$, the code distance) are fixed. 
Moreover, we construct MRD codes with different affine ranks and aperiodic MRD 
codes. 
\smallskip
\\
{\em AMS(2020)} 94B25

\endabstract

\maketitle

\section{Introduction}

Maximum rank-distance, or MRD, codes are a rank-metric analogue of MDS 
(i.e., attaining the Singleton bound) codes in the Hamming metric,
or index-$1$ orthogonal arrays.
Unlike
MDS codes,
MRD codes are shown to exist 
for all code distances and all parameters of the rank-metric space.
The first such codes were constructed by Delsarte~\cite{Delsarte:78:bilinear} 
and later by Gabidulin~\cite{Gabidulin:85:rank},
who further developed the theory of rank-metric codes,
see \cite{Gabidulin:2021}.
Much of the literature refers to that first class of MRD codes as
Gabidulin codes.

In~\cite{CKWW:2016}, de~la~Cruz et~al. studied the algebraic structure of linear MRD codes,
explained the automorphism group and asked
whether one can construct an MRD code that is not equivalent to a Gabidulin code.
In ALCOMA'2015 conference, two groups solved this problem independently.
Otal--{\"O}zbudak \cite{OtOzb:2016} solved this problem for some specific parameters.
Shekey \cite{Sheekey:2016:MRD} solved the problem for a larger class of parameters. Both solutions have similar ideas.
In~\cite{OtOzb:2017}, these results were extended to additive MRD codes.
Constructions of non-additive MRD codes   can be found in \cite{CMP:2016}, \cite{DuSi:2017}, \cite{OtOzb:2018}.
Actually, there are other constructions for linear MRD codes;
the most recent survey seems Sheekey~\cite{Sheekey:2019:MRD}.


Almost all known constructions are for $n$-by-$n$ MRD
(and mainly for linear) codes.
For $m$-by-$n$ MRD codes for $m \neq n$,
the main technique is a kind of puncturing;
a nice exception is the interleaving construction,
see e.g.~\cite{SLK:2020:Interleaved}.
For nonlinear MRD codes, there are only a few sporadic constructions
(like in~\cite{OtOzb:2018}).
We should also mention a very special case,
the $m$-by-$2$ MRD codes over $\F_2$,
which are equivalent to the transversals
in the Cayley table of the group $\Z_2^m$
(for a survey on transversals of Latin squares, see~\cite{Wanless:survey2011}).
The number of such objects
is rather well studied~\cite[Theorem~7.2]{Eberhard:more}
and applied in other combinatorial problems~\cite{PTT:bent}.

In the current paper, we construct non-linear MRD codes using switching.
Switching technique for codes and designs is well known and based on changing a part of the code (design)
in such a way that the parameters do nor change. Some special cases of switching, see e.g.~\cite{Ost:2012:switching},
are constructive and are used effectively in both computational and theoretical studies.
However, the general switching approach does not specify how to find changeable parts of codes
and needs to be concreted for each class of codes, which can be done in different ways
and result in a variety of techniques.
We use the switching technique close to that for
MDS codes studied in~\cite{Potapov:SQS-MDS}.
The key of the approach is the possibility
to switch MRD subcodes of a given MRD code,
if such subcodes exist.
We study known constructions of MRD codes for the existence
of such subcodes and propose variations of the constructions
with subcodes of different parameters.
Starting with these codes and applying switching,
we construct a large variety of MRD codes,
including codes with different values of
characteristics of non-lenearity,
such as the affine rank and
the dimension of the kernel.


The structure of the paper is as follows.
The next section contains definitions and preliminary results;
in particular, we describe switching MRD subcodes in MRD codes.
For such switching to be possible,
we need MRD codes with MRD subcodes.
Such codes are considered in Section~\ref{s:subcodes}
using twisted Gabidulin codes.
In Section~\ref{s:prod}, we consider a variant of the Cartesian product for MRD codes,
which also provides a family of codes admitting switching.
Next two sections contain important corollaries:
a lower bound on the number of MRD codes (Section~\ref{s:no})
and the possibility to construct MRD codes with different affine ranks
and MRD codes with trivial kernel (Section~\ref{s:rk}).
Section~\ref{s:concl} concludes the article.

\section{Preliminaries}

The vertex set of the \emph{bilinear forms graph}
$B_q(m,n)$ is the set of $m\times n$ matrices over $\GF{q}$;
two vertices $X$ and $Y$ are adjacent in $B_q(m,n)$
if and only if $\rank(Y-X)=1$.
The minimum-path distance in $B_q(m,n)$ coincides
with the rank distance $d(X,Y) = \rank(Y-X)$,
and the corresponding metric space is known as the
\emph{rank metric space};
we also denote it by $B_q(m,n)$ in this paper.
We assume $n\le m$, so the diameter of $B_q(m,n)$ is $n$.

If $n<m$, then the isometry group $\Aut(B_q(m,n))$ of the rank metric space (automorphism group of the bilinear forms graph) $B_q(m,n)$
is the product of $\GL(m)$ (multiplication by non-singular $m\times m$ matrices in the left),  $\GL(n)$ (multiplication by non-singular $n\times n$ matrices in the right), $\GF{q}^{mn\,+}$ (translations, addition of $m\times n$ matrices),
and $\Aut(\GF{q})$ (field automorphism, acting simultaneously in all $mn$ elements of the matrix), see e.g. \cite[Th.\,9.5.1]{Brouwer}.
Taking into account that the instances of $\GL(m)$ and $\GL(n)$ above 
intersect in a subgroup isomorphic
to $\GF{q}^\times$ (multiplication by a nonzero scalar), 
the order of $\Aut(B_q(m,n))$ is\linebreak
$[m]_q!\cdot [n]_q! \cdot (q-1)^{n+m-1} \cdot q^{\binom{m}{2}+\binom{n}{2}+mn}
\cdot \log_p q$, where $[m]_q!=\prod_{k=0}^{m-1}\sum_{t=0}^k q^t$ and $p$ is the prime divisor of $q$.
In the case $m=n$,  $\Aut(B_q(m,n))$ also includes the matrix transposing.

If, for $n'<n$, we fix (say, by $0$s) the values of the elements in $n-n'$ columns,
then we obtain a set of $q^{mn'}$ matrices that induces a subspace isometric to
(a subgraph isomorphic to) $B_q(m,n')$. This set is a maximum set of diameter $n'$
in  $B_q(m,n)$, and it will be called an \emph{$n'$-anticode}, as well as any other set
that is equivalent to it under $\Aut(B_q(m,n))$.

\begin{lemma}\label{l:anti}
 An $n'$-anticode is  a maximum set of diameter $n'$
in  $B_q(m,n)$. Moreover, every two vertices at distance at most $n'$ from each other
are included in at least one $n'$-anticode.
\end{lemma}
\begin{proof}
 Since a matrix with $n'$, $n'<n$, non-zero rows has rank at most
 $n'$, the diameter of an  $n'$-anticode is $n'$.

 The cardinality of an $n'$-anticode is obviously $q^{mn'}$.

 The Delsarte code--anticode bound \cite[Theorem~3.9]{Delsarte:1973} says that
 if $C$ is a distance-$d$
 code and $A$ is a set of diameter smaller than $d$,
 then $|C|\cdot |A|$ cannot be larger than
 the space size (the graph order).
 Since MRD codes of distance $d:=n'+1$
 (see the definition below)
 exist, e.g., Gabidulin codes, we conclude that $n'$-anticodes
 defined above are maximum sets of diameter $n'$.

 Finally, if two vertices are at distance at most $n'$ from each other,
 then their difference can be converted to a matrix with at most $n'$
 nonzero columns by nonsingular linear transformation of the column space.
 It follows that those vertices belong to the same $n'$-anticode.
\end{proof}

\begin{definition}
 A set $C$ of vertices of $B_q(m,n)$ is called a distance-$d$ \emph{MRD code}
 (maximum rank-distance code) if $|C\cap B|=1$ for every $(d-1)$-anticode $B$.
\end{definition}
By the definition, every MRD code is a diameter-perfect code, according to
\cite{AhlAydKha}. Every distance-$d$ code of cardinality $q^{m(n-d+1)}$ is MRD.

\begin{lemma}\label{l:2}
Let $S$ be the set of $q^{m'n}$ matrices of $B_q(m,n)$ with fixed $m-m'$ rows, $n\le m' \le m$
(so, the induced metric space (graph) is $B_q(m',n)$). For every $n'$-anticose $B$,
the set $B\cap S$ is either empty or an anticode in $B_q(m,n)$ corresponding to $S$.
\end{lemma}
\begin{proof}
 Up to a nonsingular linear transformation of the column space,
 we can assume that the anticode is obtained by fixing the values in
 $n-n'$ columns. The claim is then trivial.
\end{proof}

\begin{theorem}[switching MRD]\label{th:swiMRD}
 Let $S$ be the set of $q^{m'n}$ matrices of $B_q(m,n)$ with fixed (say, by $0$s)
 values in $m-m'$ rows,
 $n\le m' \le m$.
 Assume that $C$ is a distance-$d$ MRD code in $B_q(m,n)$ such that
 $|S\cap C|=q^{m'(n-d+1)}$ (i.e., $S\cap C$ is MRD in $B_q(m',n)$). Then for every other
 MRD code $R$ in $B_q(m',n)$ (corresponding to $S$), the set
 $$C_R:= C - (S\cap C) + R $$
 is an MRD code in $B_q(m,n)$ (obtained by switching).
\end{theorem}
\begin{proof}
We consider a $(d-1)$-andicode $M$. By the definition of MRD codes,
$|M\cap C|=1$. If $M\cap S=\emptyset$, then $M\cap C_R=M\cap C$ and so
$|M\cap C_R|=1$ as well. Otherwise, by Lemma~\ref{l:2} $M\cap S$ is an anticode  in $B_q(m',n)$.
Since $S\cap C$ is MRD in $B_q(m',n)$,
we have $|(M\cap S)\cap (C\cap S)|=1$,
and hence
$M\cap C=(M\cap S)\cap (C\cap S)$.
After replacing $(C\cap S)$ by $R$, we get
$M\cap C_R=(M\cap S)\cap R$.
Since $R$ is MRD in $B_q(m',n)$,
the last intersection consists of a single vertex,
and hence $|M\cap C_R|$ is also~$1$. By the definition, $C_R$ is MRD.
\end{proof}

In the next two sections we consider constructions of MRD codes that
satisfy the hypothesis of Theorem~\ref{th:swiMRD} and hence
admit switching.

\section{Gabidulin like codes with Gabidulin like subcodes}\label{s:subcodes}

In this section we construct MRD codes with MRD subcodes in the class of punctured Gabidulin twisted codes. We provide two constructions connected to two theorems. There are some limitations and some advantages of the constructions in this section. We explain these extra results in their propositions as well.
We begin with a paragraph aiming to illustrate
the main idea on a special case of classic Gabidulin codes,
without using a more advanced technique related with
the representation of matrices by linearised polynomials.

Let $n$, ${m'}$,  ${m}$,  $k$ be positive integers
such that $0<k\le n\le{m'}<m$.
In the classical construction of Gabidulin codes,
the columns of ${m'}$-by-$n$ matrices over $\GF{q}$ are identified
with elements of the field $\GF{q^{m'}}$, written in some fixed basis over $\GF{q}$.
The Gabidulin MRD code $ {C}'$ of distance $d=n-k+1$ in $B_q({m'},n)$
is defined as the linear span, over $\GF{q^{m'}}$, of the $k$ vectors
$(g_1^{q^i}, \ldots, g_n^{q^i})$, $i=0,\ldots,k-1$,
where $g_1$, \ldots, $g_n$ is an arbitrary collection
of $n$ linearly independent (over $\GF{q}$) elements of $\GF{q^{m'}}$.
If we define another Gabidulin code $ {C}$
as the span of the same basis over the larger field $\GF{q^{m}}$, $\GF{q^{m'}}\subset\GF{q^{m}}$, then we have $ {C'}\subset {C}$. With properly chosen basis of $\GF{q^{m}}$, the code ${C}$ satisfies
the hypothesis of Theorem~\ref{th:swiMRD}, which shows
that $m$-by-$n$ Gabidulin codes admit switching if $m$ has a proper divisor $m'$ such that $m'\ge n$.

To describe more general constructions, we recall the definition
of the following two functions from $\F_{q^m}$ to  $\F_q$:
$$\Norm_{q^m/q}(x)=x \cdot x^q \cdot \ldots \cdot x^{q^{m-1}}=x^{\frac{q^m-1}{q-1}},$$
$$\Tr_{q^m/q}(x)=x + x^q + \cdots + x^{q^{m-1}}.$$


\begin{construction} \label{construction1}
Let $n$, $m_1$, $m$, and $d$ be positive integers such that $2\le d\le n \le m_1 < m$
and $m_1\mid m$. Let $W_1 \subseteq \F_{q^m}$
be the kernel of the $\F_{q^{m_1}}$-linear map $\Tr_{q^m/q^{m_1}}:  \F_{q^{m}} \to \F_{q^{m_1}}$. Note that $\dim_{\F_q} W_1= m-m_1$.
Let $W$ be an $\F_q$-linear intermediate $W_1 \subseteq W \subseteq \F_{q^m}$ subspace such that $\dim W=m-n$. Put $A(x)=\prod_{w \in W}(x-w)=x^{q^{m-n}} + \mbox{m.o.t} + \alpha_0 x$,
where $\mbox{m.o.t}$ stands for middle order terms here and throughout the paper. If $d=1$, then let $\eta=0$. If $d>1$, then let $\eta \in \F_{q^{m_1}}$ such that $\Norm_{q^m/q}(\alpha_0) \neq (-1)^{m(m-d+1)} \left(\Norm_{q^{m_1}/q}(\eta)\right)^{\frac{m}{m_1}}$.
Let
$$\mathcal{P}(W,\eta)=\{A(x) \circ (a_0x + a_1 x^q + \cdots + a_{n-d} x^{q^{n-d}} + \eta a_0 x^{q^{n-d+1}}): a_0,a_1, \ldots, a_{n-d} \in \F_{q^m}\}.$$ Here and throughout the paper $\circ$ stands for the functional composition. Let $\mathcal{P}_0(W,W_1,\eta)$ be the $\F_q$-linear subspace of $\mathcal{P}(W,\eta)$ defined as
$$
\mathcal{P}_0(W,W_1,\eta)=\big\{A(x) \circ (a_0x + a_1 x^q + \cdots + a_{n-d} x^{q^{n-d}} + \eta a_0 x^{q^{n-d+1}}): a_0,a_1, \ldots, a_{n-d} \in \F_{q^{m_1}}\big\}.
$$
$\mathcal{P}(W,\eta)$ corresponds to a punctured MRD Gabidulin twisted code and $\mathcal{P}_0(W,W_1,\eta)$ corresponds to an MRD subcode of it as explained in Theorem~\ref{theorem. construction1} below.
\end{construction}

We prove that Construction~\ref{construction1} holds in the following theorem.

\begin{theorem} \label{theorem. construction1}
We keep the notation and assumptions of Construction~\ref{construction1}. We choose an ordered basis $(e_1, \ldots, e_m)$ of $\F_{q^m}$ over $\F_q$ such that the last $m-m_1$ entries form a basis of $W_1$ over $\F_q$. Note that $\Ima(A)$ is an $\F_q$-linear subspace of dimension $n$ in $\F_{q^m}$. We choose an ordered basis $(f_1, \ldots, f_n)$ of $\Ima(A)$ over $\F_q$.

For each $L \in \mathcal{P}(W,\eta)$, let $M(L)$ denote the $m\times n$ matrix representing the evaluation map $L: \F_{q^m} \to \Ima(A)$ defined as $x \mapsto L(x)$ with respect to the bases $(e_1, \ldots, e_m)$ and $(f_1, \ldots, f_n)$.
Let $\mathcal{C}(W,\eta)$ be the $\F_q$-linear subspace of $\F_{q^{m \times n}}$ consisting of $M(L)$ as $L$ runs through
$\mathcal{P}(W,\eta)$. Similarly let $\mathcal{C}_0(W,W_1,\eta)$ be the $\F_q$-linear subspace of $\mathcal{C}(W,\eta)$ consisting of $M(L)$ as $L$ runs through
$\mathcal{P}_0(W,W_1,\eta)$.

Then we have the followings:
\begin{itemize}
\item $\mathcal{C}(W,\eta)$ is an MRD code of minimum rank distance $d$.
\item  If $M \in \mathcal{C}_0(W,W_1,\eta)$, then the last $m-m_1$ rows of $M$ are the zero rows. Considering $\mathcal{C}_0(W,W_1,\eta)$ in $\F_q^{m_1 \times n}$ by ignoring the last $m-m_1$ zero rows, we obtain an MRD code of mimimum rank distance $d$.
\end{itemize}
\end{theorem}
\begin{proof}
Let $a_0,a_1, \ldots, a_{n-d} \in \F_{q^m}$, not all zero. As $A(x)=x^{q^{m-n}} + \mbox{m.o.t} + \alpha_0 x$, by definition of functional composition we have that $A(x) \circ (a_0x + a_1 x^q + \cdots + a_{n-d} x^{q^{n-d}} + \eta a_0 x^{q^{n-d+1}})$ is equal to
\be \label{ep1.theorem. construction1}
\eta^{q^{m-n}} a_0^{q^{m-n}} x^{q^{m-d+1}} + a_{n-d}^{q^{m-n}} x^{q^{m-d}} + \mbox{m.o.t} + \alpha_0 a_0 x.
\ee
Note that $A(x) \circ (a_0x + a_1 x^q + \cdots + a_{n-d} x^{q^{n-d}} + \eta a_0 x^{q^{n-d+1}})$ is an additive polynomial and its solution set in $\F_{q^m}$ is an $\F_q$-linear subspace. Using (\ref{ep1.theorem. construction1}) and \cite{Gow-Q} we conclude that its solution set has dimension at most $m-d$ over $\F_q$. This implies that the corresponding evaluation map has rank at least $d$. Consequently
$\dim_{\F_q} \mathcal{C}(W,\eta)=m(n-d+1)$. Using Singleton's bound we conclude that $\mathcal{C}(W,\eta)$ is an MRD code of minimal rank distance~$d$ in $\F_q^{m \times n}$.

Put $A_1(x)=\Tr_{q^m/q^{m_1}}(x)$. We have the property that $A_1(\alpha_1 x) =\alpha_1 A_1(x)$ for any $\alpha_1 \in \F_{q^{m_1}}$. We also have the property that $A_1(x^q)=A_1(x)^q$ for any $x \in \F_{q^m}$. Using \cite{Ore} we obtain that there exists an $\F_q$-additive polynomial $B(x) \in \F_{q^m}[x]$ such that $A(x)=B(x) \circ A_1(x)$. Hence, for $a_0,a_1, \ldots, a_{n_d} \in \F_{q^{m_1}}$ and $w_1 \in W_1$ we obtain that
\begin{multline*}
A_1(a_0w_1 + a_1 w_1^q + \cdots + a_{n-d} w_1^{q^{n-d}} + \eta a_0 w_1^{q^{n-d+1}}) \\
= a_0A_1(w_1) + a_1 A_1(w_1)^q + \cdots + a_{n-d} A_1(w_1)^{q^{n-d}} + \eta a_0 A_1(w_1)^{q^{n-d+1}}=0.
\end{multline*}
Note that $\dim_{\F_q} \mathcal{P}_0(W,W_1\eta)=m_1(n-d+1)$. Using Singleton's bound we complete the proof by showing that $\mathcal{C}_0(W,W_1,\eta)$ becomes an MRD code of minimum rank distance $d$ when considered in $\F_q^{m_1 \times n}$.
\end{proof}

It is impossible to generalize Construction~\ref{construction1} directly to the situation $m_1 \nmid m$ in some cases.
Namely we have the following proposition.
\begin{proposition} \label{proposistion.new.Wedderburn type}
As in Construction~\ref{construction1}, apart from the statement $m_1 \nmid m$, we have similar assumptions, but for some special parameters as follows. Namely
let $n$, $m_1$, $m$ be positive integers such that $n=m_1 < m$, $m_1\nmid m$ and let $d = n$. Let $W \subseteq \F_{q^m}$
be an $\F_{q^m}$-linear subspace such that $\dim_{\F_q}=m-n$. Put $A(x)=\prod_{w \in W}(x-w)=x^{q^m} + \mbox{m.o.t} + \alpha_0 x$.
Let $\eta \in \F_{q^m}$ such that $\Norm_{q^m/q}(\eta) \not \in \{(-1)^m, (-1)^m\Norm_{q^m/q}(\alpha_0)\}$.

Let $\mathcal{P}(W,\eta)=\{A(x) \circ (a_0x + \eta a_0 x^q): a_0\in \F_{q^m}\}$.

Let $(e_1, \ldots, e_n)$  be an ordered basis of $\F_{q^m}$ over $\F_q$.
Let $(f_1, \ldots, f_n)$ be an ordered basis of $\Ima(A)$ over $\F_q$.
Provided that these bases are chosen, for  each $L \in \mathcal{P}(W,\eta)$, let $M(L)$ denote the $m\times n$ matrix representing the evaluation map $L: \F_{q^m} \to \Ima(A)$ defined as $x \mapsto L(x)$ with respect to the bases $(e_1, \ldots, e_m)$ and $(f_1, \ldots, f_n)$. Let $\mathcal{C}(W,\eta)$  be the $\F_q$-linear subspace of $\F_{q^{m \times n}}$ consisting of $M(L)$ as $L$ runs through
$\mathcal{P}(W,\eta)$.

Then we have the followings:
\begin{itemize}
\item For any choice of $(e_1, \ldots, e_m)$, the set $\mathcal{C}(W,\eta)$ is an MRD code of minimum rank distance $d$.
\item It is impossible to choose an ordered basis $(e_1, \ldots, e_m)$ and a suitable $\F_q$-linear subspace $\mathcal{C}_0(W,\eta)$ of $\mathcal{C}(W,\eta)$ such that the last $m-m_1$ rows of each matrix in $\mathcal{C}_0(W,\eta)$ are the zero rows and ignoring the last $m-m_1$ zero rows, $\mathcal{C}_0(W,\eta)$ becomes an MRD code of minimum rank distance $d$.
\end{itemize}
\end{proposition}
\begin{proof}
The proof of the first item is similar to that of Theorem \ref{theorem. construction1}. It remains to prove the second item. Assume the contrary and let $(e_1, \ldots, e_m)$ be such an ordered basis. Let $W_1$ be the $\F_q$-span of the last $m-m_1$  entries in the ordered basis $(e_1, \ldots, e_m)$. Note that $\dim W_1=m-m_1$.

Let $\Psi: \F_{q^m} \to \F_{q^m}$ be the map $x \mapsto x+ \eta x^q$. As $\Norm_{q^m/q}(\eta) \neq (-1)^m$, using \cite{Gow-Q} we conclude that $\Psi$ is a permutation. Put $\widehat{W}_1=\{w_1 + \eta w_1^q: w_1 \in W_1\}$. We obtain that
$\dim_{\F_q} \widehat{W}_1=m_1$.

As $\mathcal{C}_0(W,\eta)$ corresponds to an MRD code of minimum rank distance $d$ in $\F_q^{m_1 \times m}$, using Singleton's bound we get
that $\dim_{\F_q} \mathcal{C}_0(W,\eta)=m_1$. Let $S$ be the subset of $\F_{q^m}$ consisting of $a_0$ such that
$A(x)\circ \left( a_0(x+\eta x^q)\right)$ corresponds to a matrix in $\mathcal{C}_0(W,\eta)$. This is equivalent to the combined condition that
$\dim_{\F_q} S=m_1$ and the statement that if $a_0 \in S$, then $a_0 \widehat{W}_1=W$.

As $m_1 > 0$, let $a^*_0$ be an element in $S \setminus \{0\}$. Let $\widehat{S}=\{a_0/a_0^*: a_0 \in S\}$. Note that
$\widehat{S}$ is an $\F_q$-linear subspace of $\F_{q^m}$ with the properties that $1 \in \widehat{S}$ and  the statement that if $a \in \widehat{S}$, then $a \widehat{W}_1=\widehat{W}_1$.

Let $\bar{S} \subseteq \F_{q^m}$ be the largest subset consisting of $a \in \F_{q^m}$ such that $ a\widehat{W}_1=\widehat{W}_1$. It follows from the definition that $\widehat{S} \subseteq \bar{S}$. It also follows from the definition that $\bar{S}$ is not only additive but also closed under multiplication. This implies that $\bar{S}$ is a finite integral domain and hence $\bar{S}$ is an intermediate field $\F_q \subseteq \bar{S} \subseteq \F_{q^m}$ (see also Wedderburn's Theorem, for example, in \cite[Theorem 2.55]{LN}). Put $\dim_{\F_q}\bar{S}=\rho$. As $\bar{S}$ is an intermediate field we obtain that $\rho \mid m$. Note that this also implies that $\widehat{W}_1$ is a linear space over $\bar{S}$. As $\dim_{F_q}\widehat{W}_1=m-m_1$ we conclude that $\rho \mid (m-m_1)$ and hence
\be \label{ep1.proposistion.new.Wedderburn type}
\rho \mid m_1.
\ee
Note that $\widehat{S} \subseteq \bar{S}$ and $\widehat{S}$ is  an $\F_q$-linear space of dimension $m_1$. This immediately implies that
\be \label{ep2.proposistion.new.Wedderburn type}
m_1 \le \rho.
\ee
Combining (\ref{ep1.proposistion.new.Wedderburn type}) and (\ref{ep2.proposistion.new.Wedderburn type}) we conclude that
$m_1=\rho$. This is a contradiction as $m_1 \nmid m$.
\end{proof}

In some cases it is possible to extend Construction \ref{construction1} in order to cover the case $m_1 \nmid m$. We present such an extension below.

\begin{construction} \label{construction2}
Let $n$, $m_1$, $m$, and $d$  be positive integers such that
$2 \le d \le n \le m_1 < m$.
(We do not assume $m_1 \mid m$.)
Let $\eta$ from $\F_{q^m}$ is such that $\Norm_{q^m/q}(\alpha_0) \neq (-1)^{m(m-d+1)} \Norm_{q^m/q}(\eta)$.
Assume there exist $\F_q$-linear subspaces $W,W_1 \subseteq \F_{q^m}$ and
$
S \subseteq \underbrace{\F_{q^m} \times \cdots \times \F_{q^m}}_{n-d+1~\text{times}}
$
with the following properties:
\begin{itemize}
\item[P.1.] We have $\dim_{\F_q} W=m-n$, $\dim_{\F_q} W_1=m-m_1$, and $\dim_{\F_q} S= m_1(n-d+1)$. \smallskip
\item[P.2.] For each $(a_0, \ldots, a_{n-d}) \in S$ and for each $w_1 \in W_1$ we have that
\be
a_0w_1 + a_1w_1^q + \cdots +a_{n-d}w_1^{q^{n-d}} + \eta a_0 w_1^{q^{n-d+1}} \in W.
\nn\ee
\end{itemize}
Put  $A(x)=\prod_{w \in W}(x-w)=x^{q^m} + \mbox{m.o.t} + \alpha_0 x$.
Let
$$
\mathcal{P}(W,W_1,\eta)=\{A(x) \circ (a_0x + a_1 x^q + \cdots + a_{n-d} x^{q^{n-d}} + \eta a_0 x^{q^{n-d+1}}): a_0,a_1, \ldots, a_{n-d} \in \F_{q^m}\}.
$$
Let $\mathcal{P}_0(W,W_1,\eta)$ be the $\F_q$-linear subspace of $\mathcal{P}(W,W_1,\eta)$ defined as
$$
\mathcal{P}_0(W,W_1,\eta)=\{A(x) \circ (a_0x + a_1 x^q + \cdots + a_{n-d} x^{q^{n-d}} + \eta a_0 x^{q^{n-d+1}}): (a_0,a_1, \ldots, a_{n-d}) \in S\}.
$$
As in Construction~\ref{construction1}, $\mathcal{P}(W,W_1,\eta)$ corresponds to a punctured MRD Gabidulin twisted code and $\mathcal{P}_0(W,W_1,\eta)$ corresponds to an MRD subcode of it as explained in Theorem~\ref{theorem. construction2} below.
\end{construction}

We prove that Construction~\ref{construction2} holds in the following theorem.

\begin{theorem} \label{theorem. construction2}
We keep the notation and assumptions of Construction~\ref{construction2}. We choose and ordered basis $(e_1, \ldots, e_m)$ of $\F_{q^m}$ over $\F_q$ such that the last $m-m_1$ entries form a basis of $W_1$ over $\F_q$.  We choose an ordered basis $(f_1, \ldots, f_n)$ of $\Ima(A)$ over $\F_q$.

For each $L \in \mathcal{P}(W,W_1,\eta)$, let $M(L)$ denote the $m\times n$ matrix representing the evaluation map $L: \F_{q^m} \to \Ima(A)$ defined as $x \mapsto L(x)$ with respect to the bases $(e_1, \ldots, e_m)$ and $(f_1, \ldots, f_n)$.
Let $\mathcal{C}(W,W_1,\eta)$ be the $\F_q$-linear subspace of $\F_{q^{m \times n}}$ consisting of $M(L)$ as $L$ runs through
$\mathcal{P}(W,W_1,\eta)$. Similarly let $\mathcal{C}_0(W,W_1,\eta)$ be the $\F_q$-linear subspace of $\mathcal{C}(W,W_1,\eta)$ consisting of $M(L)$ as $L$ runs through
$\mathcal{P}_0(W,W_1,\eta)$.

Then we have the followings:
\begin{itemize}
\item $\mathcal{C}(W,W_1,\eta)$ is an MRD code of minimum rank distance $d$.
\item  If $M \in \mathcal{C}_0(W,W_1,\eta)$, then the last $m-m_1$ rows of $M$ are the zero rows. Considering $\mathcal{C}_0(W,W_1,\eta)$ in $\F_q^{m_1 \times n}$ by ignoring the last $m-m_1$ zero rows, we obtain an MRD code of mimimum rank distance $d$.
\end{itemize}
Moreover Construction \ref{construction2} generalizes Construction \ref{construction1}.
\end{theorem}
\begin{proof}
If $a_0, \ldots, a_{n-d} \in \F_{q^m}$, not all zero, then the same arguments of the proof of Theorem \ref{theorem. construction1} imply that the rank of the map $A(x) \circ (a_0x + a_1 x^q + \cdots + a_{n-d} x^{q^{n-d}} + \eta a_0 x^{q^{n-d+1}})$ is at least $d$. Similarly we obtain that $\dim_{\F_q} \mathcal{C}(W,W_1,\eta)=m(n-d+1)$. Using Singleton's bound we conclude that $\mathcal{C}(W,W_1,\eta)$ is an MRD code of minimum rank distance $d$.

If $w_1 \in W_1$ and $(a_0, \ldots, a_{n-d}) \in S$, then using property P.2 of Construction \ref{construction2} we
obtain that $A(x) \circ (a_0w_1 + a_1 w_1^q + \cdots + a_{n-d} w_1^{q^{n-d}} + \eta a_0 w_1^{q^{n-d+1}})=0$. As $\dim_{\F_{q}} S=m_1(n-d+1)$ we complete the proof.
\end{proof}

Next we present an application of a slight modification of Construction~\ref{construction2}. The following result cannot be obtained using Construction~\ref{construction1} or the product construction of Section~\ref{s:prod} below. Moreover for any given integer $\mu$, it is a construction of an MRD code with the minimum rank distance $\mu$ having a suitable subcode, whose elements have zero rows at the end and the subcode is still MRD of with minimum rank distance $\mu$ after ignoring these zero rows.

\begin{proposition} \label{proposition.subtract-many}
Let $\mu \ge 1$ and $\ell \ge 2$ be integers. Put $n=\mu(\ell-1)$ and $m_1=\mu \ell$ so that $n<m$. Moreover put $\mu_1=\mu$.
Let $W \subseteq \F_{q^{\mu}}$. Let $W_1=\{w_1 \in \F_{q^m}: \Tr_{{q^m}/{q^{\mu}}}(w_1)=0\}$. Note that $W$, $W_1$ are $\F_q$-linear subspaces with $\dim_{\F_q} W=\mu$ and $\dim_{\F_q} W_1=m-\mu$. Let $S$ be the set defined as
\be \label{e1.proposition.subtract-many}
\begin{array}{r@{}c@{}l}
S &\mbox{}= \Big\{&  (a_0, a_1, \ldots, a_{\mu(\ell-2)}) \in \F_{q^m} \times \cdots \times \F_{q^m}  \\
& & \hspace{0.2cm} a_i=0 \; \mbox{if $\mu \nmid i$}, \\
& & \hspace{0.2cm} a_0 \in W_1, \\
& & \hspace{0.2cm} a_{\mu}=a_0 + a_0^{q^\mu}, \\
& & \hspace{0.2cm} a_{2\mu}=a_0 + a_0^{q^\mu} + a_0^{q^{2\mu}}, \\
& & \hspace{1cm}\vdots \\
& &   \hspace{0.2cm} a_{\mu(\ell-2)}=a_0 + a_0^{q^\mu} + a_0^{q^{2\mu}} + \cdots + a_0^{q^{(\ell -2)\mu}}\Big\}.
\end{array}
\ee
Let
$$
\mathcal{P}(W)=\{A(x) \circ (a_0x + a_1 x^q + \cdots + a_{\mu(\ell-2)} x^{q^{\mu(\ell-2)}} : a_0,a_1, \ldots, a_{\mu(\ell-2)} \in \F_{q^m}\} .
$$
Let $\mathcal{P}_0(W,W_1,S)$
be the $\F_q$-linear subspace of $\mathcal{P}(W)$ defined as
$$
\mathcal{P}_0(W,W_1,S)=\{A(x) \circ (a_0x + a_1 x^q + \cdots + a_{\mu(\ell-2)} x^{q^{\mu(\ell-2)}}: (a_0,a_1, \ldots, a_{\mu(\ell-2)}) \in S\}.
$$
We choose an ordered basis $(e_1, \ldots, e_m)$ of $\F_{q^m}$ over $\F_q$ such that the last $m-m_1$ entries form a basis of $W_1$ over $\F_q$.  We choose an ordered basis $(f_1, \ldots, f_n)$ of $\Ima(A)$ over $\F_q$.

For each $L \in \mathcal{P}(W)$, let $M(L)$ denote the $m\times n$ matrix representing the evaluation map $L: \F_{q^m} \to \Ima(A)$ defined as $x \mapsto L(x)$ with respect to the bases $(e_1, \ldots, e_m)$ and $(f_1, \ldots, f_n)$.
Let $\mathcal{C}(W)$ be the $\F_q$-linear subspace of $\F_{q^{m \times n}}$ consisting of $M(L)$ as $L$ runs through
$\mathcal{P}(W)$. Similarly let $\mathcal{C}_0(W,W_1,S)$ be the $\F_q$-linear subspace of $\mathcal{C}(W)$ consisting of $M(L)$ as $L$ runs through
$\mathcal{P}_0(W,W_1,S)$.

Then we have the followings:
\begin{itemize}
\item $\mathcal{C}(W)$ is an MRD code of minimum rank distance $\mu$. \smallskip
\item  If $M \in \mathcal{C}_0(W,W_1,S)$, then the last $m-m_1$ rows of $M$ are the zero rows. Considering $\mathcal{C}_0(W,W_1,S)$ in $\F_q^{m_1 \times n}$ by ignoring the last $m-m_1$ zero rows, we obtain an MRD code of minimum rank distance $\mu$.
\end{itemize}
Note that $\mathcal{C}(W)$ is in $\F_q^{m \times n}$  with $m > n$, while the subcode $\mathcal{C}_0(W,W_1,S)$ is considered in $\F_q^{m_1 \times n}$ with $m_1< n$. This property is different from the earlier results in this section.
\end{proposition}
\begin{proof}
As $m > n$, an MRD code of rank $\mu$ in $\F_q^{m \times n}$ has $\F_q$-dimension $m(n\mu+1)$. This is the same as $\dim_{\F_q} \mathcal{P}(W)$ as $a_i \in \F_{q^m}$ in the definition and $\mathcal{P}(W)$ and the number of such coefficients is $1+(\ell -2)=1+ (n-\mu)$.

As $m_1< n$, an MRD code of rank $\mu$ in $\F_q^{m_1 \times n}$ has $\F_q$-dimension $n(m_1-\mu+1)=n=\mu(\ell-1)$. This is the same as $\dim_{\F_q} W_1$. Moreover if $(a_0, a_1, \ldots, a_{\mu(\ell-2)}) \in S$, then we have
\begin{multline} \label{ep1.proposition.subtract-many}
\big(x^{q^\mu}-x \big) \circ \Big( a_0x + a_{\mu} x^{q^{\mu}} + a_{2\mu} x^{q^{2\mu}} + \cdots + a_{(\ell-2)\mu} x^{q^{(\ell-2)\mu}}\Big ) \\
=-a_0 \big(x + x^{q^{\mu}} + \cdots + x^{q^{(\ell-1)\mu}}\big),
\end{multline}
where we use the properties of the coefficients given in (\ref{e1.proposition.subtract-many}). Hence evaluating the function in (\ref{ep1.proposition.subtract-many}) at $w_1 \in W_1$ we obtain $0$. This completes the proof.
\end{proof}

\section{Product construction}\label{s:prod}

In this section, we describe the direct product construction for MRD codes.
It is an analog of a similar construction for MDS codes (in particular, for latin squares),
which constructs an MDS code over a $q'q''$-ary alphabet
from MDS codes of the same length and distance over a $q'$-ary alphabet
and a $q''$-ary alphabet.
Actually, MRD codes in $B_q(m,n)$, $n\le m$, form a subclass of $q^m$-ary MDS codes
if we treat columns of the matrices as symbols of the corresponding $q^m$-ary alphabet.
It happens that the direct product of two such codes
(in the manner that corresponds to the direct product of the alphabets)
keeps not only the MDS property, but also the MRD one.

\begin{remark}
The product construction we consider is also known as vertical interleaving;
for MRD codes, it was firstly suggested in 
\cite{LoidreauOverbeck:2006}.
The most general form was considered in~\cite{SLK:2020:Interleaved},
where the height of the matrices of the component codes are not required to be equal.
\end{remark}

For two codes $C'$ and $C''$ in $B_q(m',n)$ and $B_q(m'',n)$,
respectively, $n\le m'$, $n\le m''$,
we denote
$$
C' \times C'' = \left\{ \genfrac{[}{]}{0pt}{0}{X'}{X''}:\ X'\in C',\ X''\in C'' \right\} .
$$
\begin{proposition}[direct product]\label{p:product}
 If $C'$ is a distance-$d$ MRD code in $B_q(m',n)$, $n\le m'$,
 and $C''$ is a distance-$d$ MRD code in $B_q(m'',n)$, $n\le m''$,
 then $C = C' \times C''$ is a distance-$d$ MRD code in $B_q(m'+m'',n)$.
 If, additionally, $C'$ and $C''$ are linear (over $\FF_q$), then  $C' \times C''$ is linear too.
\end{proposition}
\begin{proof}
 If $X=\genfrac{[}{]}{0pt}{1}{X'}{X''}$ and $Y=\genfrac{[}{]}{0pt}{1}{Y'}{Y''}$ are two different codewords
 of  $C$, then $X'\ne Y'$ or  $X''\ne Y''$.
 If $X''\ne Y''$, then $\rank(Y''-X'')\ge d$
 and hence $\rank(Y-X)\ge d$, which confirms the code distance of $C$.
  If $X''= Y''$, then $X'\ne Y'$, $\rank(Y'-X')\ge d$,
 and again $\rank(Y-X)\ge d$. The linearity is straightforward.
 \end{proof}
For the MRD code $C$ constructed as above, we observe the following.
For given $X''$ in $C''$, the set $S$ of matrices
whose last $m''$ rows coincide with $X''$ satisfy the hypothesis
of Theorem~\ref{th:swiMRD}. I.e., $S\cap C$ is MRD in $B_q(m',q)$.
Such MRD subcode (actually, it is a coset of $C'$) can be switched independently
for each $X''$. This results in the following switching variation
of the product construction.
\begin{theorem}\label{th:product}
 Let $m'$, $m''$, $n$, $d$ satisfy $1\le d\le n\le m'$, $n\le m''$.
 Let $C''$ be a distance-$d$ MRD code in $B_q(m'',n)$,
 and for every $X''$ in $C''$ let
 $C'_{X''}$ be a distance-$d$ MRD code in $B_q(m',n)$.
 Then the code $C$ defined as
 $$C = \left\{ \genfrac{[}{]}{0pt}{0}{X'}{X''}:\ X''\in C',\ X'\in C'_{X''} \right\}$$
 is a distance-$d$ MRD code in $B_q(m'+m'',n)$.
 \end{theorem}
The proof is the same as for Proposition~\ref{p:product}, taking into account that
$X'$ and $Y'$ belong to the same code $C'_{X''}$ if $X''=Y''$.
Another interesting corollary of the product construction is constructing new linear MRD codes.
The following fact is straightforward from
Proposition~\ref{proposistion.new.Wedderburn type}.
\begin{corollary}\label{c:product}
Assume, under the notation and the hypothesis of Proposition~\ref{p:product},
that $C'$ and $C''$ are $\FF_q$-linear, $m'=n=d$,
and $m'$ does not divide $m''$.
Then the code $C=C'\times C''$
is not equivalent to any Gabidulin or twisted Gabidulin code.
\end{corollary}

\begin{proof}
By construction, the product MRD code has an MRD subcode
consisting of $n$-by-$n$ matrices. The codes considered in~Proposition~\ref{proposistion.new.Wedderburn type}
have no such subcodes.
\end{proof}

\section{On the number of inequivalent codes}\label{s:no}

If the hypothesis of Theorem~\ref{th:swiMRD}
is satisfied
(which is true for twisted Gabidulin codes
if $m'$ divides $m$ and for product MRD codes
for any parameters satisfying $2\le d\le n\le m'\le m-n$), then
the code $C$ is partitioned into
$q^{(n-d+1)(m-m')}$ cosets of $S\cap C$.
Each of these cosets can be independently
replaced to a coset of $R$
(or any other MRD code with the same parameters),
resulting in a new
MRD code with the the same parameters as $C$.
So, we can obtain more than
$$
2^{q^{(n-d+1)(m-m')}}
$$
different MRD codes.
To evaluate the number of inequivalent codes,
we can divide this number by the maximum size of an equivalence class,
which is not more than the number of automorphisms of $B_q(m,n)$.
So, we have at least
$$
\frac{2^{q^{(n-d+1)(m-m')}}}{
     |\mathrm{GL}(m)|
\cdot|\mathrm{GL}(n)|
\cdot|\mathbb F_q^{mn+}|
\cdot|\mathrm{Aut}(\mathbb{F}_q)|}
\ge \frac{2^{q^{(n-d+1)(m-m')}}}{
 q^{m^2} \cdot q^{n^2} \cdot q^{mn} \cdot \log q
} =
2^{q^{(n-d+1)m - o(m)}}
$$
inequivalent MRD codes
as $n$ and $m'$ are fixed and $m\to\infty$.
Of course, most of them are non-linear.
\begin{corollary}\label{number}
 If $d \le n \le m$,
 then there are at least $2^{q^{(n-d+1-o(1))m}}$
 inequivalent MRD codes of distance $d$
 in $B_q(m,n)$ as $m\to\infty$.
\end{corollary}

\section{Codes with different affine ranks and aperiodic codes}\label{s:rk}

With the switching approach, we can construct many MRD codes
with the same parameters, most of the codes being non-linear.
Two known numerical
characteristics of codes that show
how far a code is from being linear
are the affine rank and the kernel size,
see the definitions below.
In this section, we will show how to construct MRD codes with different affine ranks
and MRD codes with trivial kernel.

\subsection{Codes with different affine ranks}
The \emph{affine rank} of a code $C$ in a vector space
(in our case, the set of all $m$-by-$n$ matrices with elements from $\FF_q$,
forming an $mn$-dimensional vector space over $\FF_q$)
is the dimension of the affine span of $C$.
Note that this concept does not depend on the metric,
and the words `rank' has different meanings
in `affine rank' and `rank distance'.
If the code is linear, or a coset of a linear code,
then its affine rank coincides with its dimension.

\begin{theorem}[on the affine rank of MRD codes]
 If $1<d\le n \le  m/2$ and $q$ is a prime power,
 then in $B_q(m,n)$ there are
 distance-$d$ MRD codes of each affine rank from
 $m(n-d+1)$ to $m(n-d+1) + \rho$, where
$$
\rho = \max_{m' \in \{n,\ldots,m-n\} }
\min \{ m'(d-1),\ q^{(m-m')(n-d+1)}-(m-m')(n-d+1)-1 \}.
$$
\end{theorem}
\begin{proof}
For each $m'$ in the range $\{n,\ldots,m-n\}$,
 there is a product (see Section~\ref{s:prod})
 $\F_q$-linear distance-$d$ MRD code
 $C$ in $B_q(m,n)$ that is representable as the union of
 cosets of a
 $\F_q$-linear distance-$d$ MRD code
 $M$ in $B_q(m',n)$ (regarded as a subgraph of $B_q(m,n)$):
 \begin{equation}\label{eq:MMMM}
  C = \bigcup_{x\in P} (M+x),
 \end{equation}
 where $P$ is a set of representatives,
 $|P|=q^{(m-m')(n-d+1)}$,
 $\overline 0 \in P$.
 By Theorem~\ref{th:swiMRD},
 every instance of $M$ in~\eqref{eq:MMMM} can be replaced by another
 MRD code in $B_q(m',n)$, in particular, by a translation
 $M+y_x$ of $M$ with any vector $y_x$ from $B_q(m',n)$:
 \begin{equation}\label{eq:MMMMy}
  C' = \bigcup_{x\in P} (M+y_x+x).
 \end{equation}
 Our goal is to choose $y_x$ in such a way that
 $C'$ has the required affine rank.

 We now choose two collections of vectors. The first collection is a subset
 $P'$ of $P$ such that $|P'|=(m-m')(n-d+1)$ and the linear span
 $\langle M,P' \rangle$ coincides with $C$ (in particular,
 the affine span of $M\cup P'\cup\{\overline0\}$ is $C$).

 The second collection $Z=\{z_1,\ldots,z_{m'(d-1)}\}$
 of vectors such that the linear span
 $\langle M,Z \rangle$ exhausts the vertex set of $B_q(m',n)$
 (this is possible because $m'n = \dim(M)+m'(d-1)$).

 In $P\backslash P' \backslash \{\overline0\}$, we choose
 $k$ vectors $x_1$, \ldots, $x_k$, where $k\le \min(|P|-|P'|-1,{m'(d-1)})$
 and set $y_{x_i}:=z_i$, $i=1,\ldots,k$.
 For every other vector $x$ from $P\backslash\{x_1, \ldots, x_k\}$,
 we set $y_{x_i}:=\overline0$.
 We now see that the affine span of $C'$ includes the linear code $C$
 and, additionally, the vectors $z_1$, \ldots, $z_k$,
 each increasing the affine rank by $1$.
 So, the $\mathrm{affine\_rank}(C')=\dim(C)+k$,
 $k\in \{0,\ldots, \min(|P|-|P'|-1,{m'(d-1)})$,
 where
 $|P|-|P'|-1 = q^{(m-m')(n-d+1)} - (m-m')(n-d+1) -1$.
 It remains to note that we can choose the best result over all $m'$ in $\{n,\ldots,m-n\}$.
\end{proof}

\begin{remark}\label{r:rank}
The goal of this section was to demonstrate,
with reasonably simple arguments,
the possibility to construct MRD codes
of the same parameters with wide variety of affine ranks.
We are convinced that for most of parameters higher values of the affine rank can be achieved
with further developing the technique.
In particular, in the proof above,
we consider only a very special
variant of switching,
where each coset of the MRD subcode $M$
is replaced by its translation. As a result, every such replacement contributes at most $+1$ to the affine rank.
This approach can be improved if one replace a coset of $M$
by a coset of another MRD code $M'$ in the same $B_q(m',n)$ subgraph.
Depending on the dimension of the linear span
$\langle M \cup M' \rangle$, each such replacement can add more than $1$
to the affine rank of the resulting code.
With such generalized approach, the total addition $\rho$ to the affine rank is still bounded
by the value $(m-n)(d-1)$. Further increasing of this value can be hypothetically
related with switching MRD subcodes in ``non-parallel'' $B_q(m',m)$ subgraphs, see the brief discussion in the conclusion part of the paper.
\end{remark}

\subsection{Kernel}
The \emph{kernel} of a code is the set of its periods:
$$ \mathrm{ker}(C) = \{x :\, C+x=C \}.$$

For a linear code $C$,
it holds $\mathrm{ker}(C) = C$.
Inversely, if $\mathrm{ker}(C) = C$
then $C$ is additive (closed under addition),
but, for non-prime $q$,
is not necessarily linear.
If $\mathrm{ker}(C)$ consists of
only the zero vector,
then
the code is said to have \emph{trivial kernel},
or just to be \emph{aperiodic}.

\begin{lemma}\label{l:without}
 Let $M$ be a linear
 rank-distance-$d$ code
 in $B_q(m,n)$, $d>1$, $|M|>1$,
 and let $\bar x$ be a non-zero
 matrix in $B_q(m,n)$.
 Then there is a linear
 rank-distance-$d$ code $D$
 in $B_q(m,n)$ such that $|M|=|D|$
 and $\bar x \not \in D$.
 \end{lemma}
\begin{proof}
 If $\bar x \not \in D$, then we take
 $D=M$. Otherwise, let $i$ be the index
 of any non-zero row of $\bar x$
 and $j$ any other row index.
 Denote by $\pi_{i,j}$ the transformation
 of matrices that add the $i$th row
 to the $j$th row. Obviously,
 $\pi_{i,j}$ is an isometry of $B_q(m,n)$,
 and hence $D:=\pi_{i,j}(M)$ has the same
 code parameters as $M$.
 It remains to observe that the rank
 distance between
 $\bar x$ and $\pi_{i,j}(\bar x)$
 is $1$; hence $D$,
 containing $\pi_{i,j}(\bar x)$,
 does not contain $\bar x$.
\end{proof}

\begin{lemma}\label{l:cup}
 Let $M$ be a linear
 rank-distance-$d$ code
 in $B_q(m,n)$, $d>1$, $\dim(M)=r$.
 Then for every $k\in \{0,\ldots,r\}$,
 there is a collection
 $\{M=M_0,M_1,\ldots,M_k\}$
 of codes of the same cardinality
 and code distance such that
 $\dim(M_0\cap\ldots\cap M_k)\le r-k$.
 Moreover, if $k>0$, then $M_1 \not \in \{M_0,M_2,\ldots,M_k\}$.
 \end{lemma}
\begin{proof}
We proceed by induction in $k$.
If $k=0$, the claim is trivial.
If $k> 0$, then from the inductive
hypothesis we have a collection
 $\{M=M_0,M_1,\ldots,M_{k-1}\}$
 such that
 the dimension of
 $U:=M_0\cap\ldots\cap M_{k-1}$ is at most
 $r-k+1$. If this dimension is $0$,
 we take $M_k=M$.
 Otherwise, we choose a nonzero matrix
 $\bar x$ in $U$ and take $M_k=D$,
 where $D$ is from
 Lemma~\ref{l:without}
 (since $\bar x$ belongs to all $M_i$, $i\in\{0,\ldots,k-1\}$
 but not to $M_k$, all $k$ codes are distinct, which provides the last claim of the lemma).
\end{proof}

\begin{theorem}[aperiodic MRD codes]
 If $1<d\le n \le m/2$ and $q$ is a prime power,
 then there are aperiodic
 distance-$d$ MRD codes in $B_q(m,n)$.
\end{theorem}
\begin{proof}
 There is an $\F_q$-linear distance-$d$ MRD code
 $C$ in $B_q(m,n)$ that is representable as the union of
 cosets of a
 $\F_q$-linear distance-$d$ MRD code
 $M$ in $B_q(n,n)$ (regarded as a subgraph of $B_q(m,n)$):
 \begin{equation}\label{eq:MMMM_}
  C = \bigcup_{x\in P} (M+x),
 \end{equation}
 where $P$ is a set of representatives,
 $|P|=q^{(m-n)(n-d+1)}> n(n-d-1)$,
 $\overline 0 \in P$.
 By Theorem~\ref{th:swiMRD},
 for every $x$ in $P$, the corresponding
 instance of $M$ in~\eqref{eq:MMMM_}
 can be replaced by another
 MRD code $M^{(x)}$ in $B_q(n,n)$ with the same code parameters as $M$
 \begin{equation}\label{eq:MMMM_y}
  C' = \bigcup_{x\in P} M^{(x)}.
 \end{equation}
 We choose distinct $x_0$, \ldots, $x_r$ in $P$, where $r=n(n-d-1)=\dim(M)$,
 and set $M^{(x_i)}:=M_i$, $i=0,\ldots,r$, where
 $M_i$ are from Lemma~\ref{l:cup} $(k=r)$.
 For the remaining $|P|-r-1$ elements $x$ of $P$,
 we keep $M^{(x_i)}:=M$.

 We now consider an arbitrary non-zero $y$ in $B_q(m,n)$ and
  show that it is not a period of $C'$.
  Seeking a contradiction, assume $y$ is a period.
  In this case, both $x_1$ and $x_1+y$ are in $C'$,
  and we consider two subcases.

  (i) If $x_1+y \in x_1+M_1$, then we also have
  $x_i+y \in x_i+M_i$, $i=0,\ldots,r$.
  This implies $y=\bar 0$ because
  $M_0\cap\ldots\cap M_d=\{\bar 0\}$ by Lemma~\ref{l:cup}.

  (ii) If $x_1+y \in x+M^{(x)}$ for some $x$
  from $P\backslash\{x_1\}$, then
  $x_1+M_1 + y = x+M^{(x)}$ and hence $M_1 = M^{(x)}$,
  which contradicts the last claim of  Lemma~\ref{l:cup}.
\end{proof}

\section{Conclusion}\label{s:concl}

We have established a lower bound
$2^{q^{(n-d+1-o(n))m}}$ on the number of MRD codes
of distance $d$ in $B_q(m,n)$
as $m$ grows.
A trivial upper bound on the number of MRD codes is
$$(q^{(d-1)m})^{q^{(n-d+1)m}}
 = 2^{m\cdot q^{(n-d+1)m}\cdot (1+o(1)) }$$
(indeed, the space is partitioned into $q^{(n-d+1)m}$
anticodes of cardinality $q^{(d-1)m}$, each containing
exactly one codeword).
Hence, the logarithm of this number is in
$O(m\cdot q^{(n-d+1)m})$, $m\to\infty$,
while our lower bound on this logarithm
is only $\mathrm{const}\cdot q^{\mathrm(n-d+1)m}$.
With the switching approach,
further improvement of the lower bound
is possible if one can switch ``non-parallel'' MRD subcodes
of the initial MRD code. This strategy worked well
for evaluating the number of $Q$-ary non-linear MDS codes
(equivalently, systems of
$k$ strongly orthogonal latin $m$-cubes of order
$Q$, where $k$ and $m$ are constants and $Q$ grows)
of fixed length and distance and growing alphabet size $Q$
(an analog of our $q^m$), see~\cite{Potapov:SQS-MDS}.
To apply a similar switching strategy to constructing MRD codes,
we need to start with an MRD
code with many different $\GF{q}$-linear MRD subcodes,
corresponding to different, ``non-parallel'' $B_q(m',n)$-subgraphs.
Developing this idea is one of perspective directions
for the further study.

Another point we want to mention is that 
the subcode switching for MRD codes has some restrictions: 
with this method, we could only construct new MRD codes 
with parameters satisfying $m\ge 2n$.
This is enough to obtain powerful asymptotic results, but do not allow to construct
new codes with $n\le m<2n$.
This restriction can be tried to bypass with other variations of switching technique.
The most common way to construct codes by switching in some direction,
see e.g.~\cite{Ost:2012:switching},
can also be applied to MRD codes: a code $C'$ is obtained from a code $C$
by switching in direction $v$, where $v$ is a vector from the ambient vector space,
if $C'$ is a subset of $C \cup (C+v)$. For given $C$ and $v$,
finding all such $C'$ is a simple computational problem, which is good for
experiments with small-parameter codes. However, developing this technique theoretically
for arbitrary parameters remains an open challenging problem for MRD codes. 
One of the main questions here is the following: 
given a starting linear code $C$ and a vector $v$ 
(often, of weight $1$, although this is not a necessary restriction),
what is the minimal subspace $S$ of $C$ such that $C \cup (S+v) \setminus S$
has the same parameters as $C$?
More generally, one can ask what is 
the minimum difference between two MRD codes
of the same parameters. This question is common for many classes of optimal codes,
and one of the developed methods to make lower bounds for this minimum is to study
eigenspaces of the ambient graph, see~\cite{SotVal:survey2021} for a survey for different graphs
and \cite{Sotnikova:bil} for the bilinear form graphs.



\begin{thebibliography}{10}

\bibitem{AhlAydKha}
R.~Ahlswede, H.~K. Aydinian, and L.~H. Khachatrian.
\newblock On perfect codes and related concepts.
\newblock {\em \href{http://link.springer.com/journal/10623}{Des. Codes
  Cryptography}}, 22(3):221--237, 2001.
\newblock \DOI{10.1023/A:1008394205999}.

\bibitem{Brouwer}
A.~E. Brouwer, A.~M. Cohen, and A.~Neumaier.
\newblock {\em Distance-Regular Graphs}.
\newblock Springer-Verlag, Berlin, 1989.
\newblock \DOI{10.1007/978-3-642-74341-2}.

\bibitem{CMP:2016}
A.~{Cossidente}, G.~{Marino}, and F.~{Pavese}.
\newblock Non-linear maximum rank distance codes.
\newblock {\em \href{http://link.springer.com/journal/10623}{Des. Codes
  Cryptography}}, 79(3):597--609, 2016.
\newblock \DOI{10.1007/s10623-015-0108-0}.

\bibitem{CKWW:2016}
J.~de~la Cruz, M.~Kiermaier, A.~Wassermann, and Willems W.
\newblock Algebraic structures of {MRD} codes.
\newblock {\em \href{http://aimsciences.org/journals/amc/}{Adv. Math.
  Commun.}}, 10(3):499--510, Aug. 2016.
\newblock \DOI{10.3934/amc.2016021}.

\bibitem{Delsarte:1973}
P.~Delsarte.
\newblock {\em An Algebraic Approach to Association Schemes of Coding Theory},
  volume~10 of {\em Philips Res. Rep., Supplement}.
\newblock N.V.~Philips' Gloeilampenfabrieken, Eindhoven, 1973.

\bibitem{Delsarte:78:bilinear}
P.~Delsarte.
\newblock Bilinear forms over a finite field, with applications to coding
  theory.
\newblock {\em \href{http://www.sciencedirect.com/science/journal/00973165}{J.
  Comb. Theory, Ser.~A}}, 25:226--241, Nov. 1978.
\newblock \DOI{10.1016/0097-3165(78)90015-8}.

\bibitem{DuSi:2017}
N.~Durante and A.~Siciliano.
\newblock Non-linear maximum rank distance codes in the cyclic model for the
  field reduction of finite geometries.
\newblock {\em \href{http://www.combinatorics.org}{Electr. J. Comb.}},
  24(2):\#2.33(1--18), June 2017.
\newblock \DOI{10.37236/6106}.

\bibitem{Eberhard:more}
S.~Eberhard.
\newblock More on additive triples of bijections.
\newblock E-print 1704.02407, arXiv.org, 2017.
\newblock Available at \url{http://arxiv.org/abs/1704.02407}.

\bibitem{Gabidulin:85:rank}
{\`E}.~M. Gabidulin.
\newblock Theory of codes with maximum rank distance.
\newblock {\em \href{http://link.springer.com/journal/11122}{Probl. Inf.
  Transm.}}, 21(1):1--12, 1985.
\newblock translation from
  \href{http://www.mathnet.ru/php/journal.phtml?jrnid=ppi\&option_lang=eng}{Probl.
  Peredachi Inf.} 21:1 (1985), 3--16.

\bibitem{Gabidulin:2021}
E.~M. Gabidulin.
\newblock {\em Rank Codes}.
\newblock TUM.University Press, Munich, 2021.
\newblock \DOI{10.14459/2021md1601193}.

\bibitem{Gow-Q}
R.~{Gow} and R.~{Quinlan}.
\newblock Galois theory and linear algebra.
\newblock {\em
  \href{http://www.sciencedirect.com/science/journal/00243795}{Linear Algebra
  Appl.}}, 430(7):1778--1789, 2009.
\newblock \DOI{10.1016/j.laa.2008.06.030}.

\bibitem{LN}
R.~Lidl and H.~Niederreiter.
\newblock {\em Finite Fields}, volume~20 of {\em Encycl. Math. Appl.}
\newblock Cambridge Univ. Press, Cambridge, 1997.

\bibitem{LoidreauOverbeck:2006}
P.~Loidreau and R.~Overbeck.
\newblock Decoding rank errors beyond the error correction capability.
\newblock In {\em Proc. {T}enth Int. Workshop on Algebraic and Combinatorial
  Coding Theory {ACCT-10}}, pages 186--190, Zvenigorod, Russia, September 2006.

\bibitem{Ore}
O.~Ore.
\newblock On a special class of polynomials.
\newblock {\em \href{http://www.ams.org/tran/}{Trans. Am. Math. Soc.}},
  35:559--584, 1933.
\newblock \DOI{10.2307/1989849}.

\bibitem{Ost:2012:switching}
P.~R.~J. {\"O}sterg{\aa}rd.
\newblock Switching codes and designs.
\newblock {\em
  \href{http://www.sciencedirect.com/science/journal/0012365X}{Discrete
  Math.}}, 312(3):621--632, 2012.
\newblock \DOI{10.1016/j.disc.2011.05.016}.

\bibitem{OtOzb:2016}
K.~Otal and F.~{\"O}zbudak.
\newblock Explicit constructions of some non-{G}abidulin linear maximum rank
  distance codes.
\newblock {\em \href{http://aimsciences.org/journals/amc/}{Adv. Math.
  Commun.}}, 10(3):589--600, Aug. 2016.
\newblock \DOI{10.3934/amc.2016028}.

\bibitem{OtOzb:2017}
K.~Otal and F.~{\"O}zbudak.
\newblock Additive rank metric codes.
\newblock {\em
  \href{http://ieeexplore.ieee.org/xpl/RecentIssue.jsp?punumber=18}{IEEE Trans.
  Inf. Theory}}, 63(1):164--168, Jan. 2017.
\newblock \DOI{10.1109/TIT.2016.2622277}.

\bibitem{OtOzb:2018}
K.~Otal and F.~{\"O}zbudak.
\newblock Some new non-additive maximum rank distance codes.
\newblock {\em
  \href{http://www.sciencedirect.com/science/journal/10715797}{Finite Fields
  Appl.}}, 50:293--303, March 2018.
\newblock \DOI{10.1016/j.ffa.2017.12.003}.

\bibitem{Potapov:SQS-MDS}
V.~N. Potapov.
\newblock On the number of {SQS}s, latin hypercubes and {MDS} codes.
\newblock {\em
  \href{http://onlinelibrary.wiley.com/journal/10.1002/(ISSN)1520-6610}{J.
  Comb. Des.}}, 26(5):237--248, 2018.
\newblock \DOI{10.1002/jcd.21603}.

\bibitem{PTT:bent}
V.~N. Potapov, A.~A. Taranenko, and {Yu}.~V. Tarannikov.
\newblock Asymptotic bounds on numbers of bent functions and partitions of the
  {B}oolean hypercube into linear and affine subspaces.
\newblock E-print 2108.00232, arXiv.org, 2021.
\newblock \DOI{10.48550/arXiv.2108.00232}.

\bibitem{Sheekey:2016:MRD}
J.~Sheekey.
\newblock A new family of linear maximum rank distance codes.
\newblock {\em Advances in Mathematics of Communications}, 10(3):475--488,
  2016.
\newblock \DOI{10.3934/amc.2016019}.

\bibitem{Sheekey:2019:MRD}
J.~Sheekey.
\newblock {MRD} codes: Constructions and connections.
\newblock volume~23 of {\em Radon Ser. Comput. Appl. Math.}, pages 255--285. De
  Gruyter, Berlin, 2019.
\newblock \DOI{10.1515/9783110642094-013}.

\bibitem{SLK:2020:Interleaved}
V.~Sidorenko, V.~Li, and G.~Kramer.
\newblock On interleaved rank metric codes.
\newblock In {\em 2020 Algebraic and Combinatorial Coding Theory (ACCT)}, pages
  128--134, 2020.
\newblock \DOI{10.1109/ACCT51235.2020.9383406}.

\bibitem{SotVal:survey2021}
E.~Sotnikova and A.~Valyuzhenich.
\newblock Minimum supports of eigenfunctions of graphs: a survey.
\newblock {\em The Art of Discrete and Applied Mathematics}, pages
  P2.09(1--34), 2021.
\newblock \DOI{10.26493/2590-9770.1404.61e}.

\bibitem{Sotnikova:bil}
E.~V. Sotnikova.
\newblock Minimum supports of eigenfunctions in bilinear forms graphs.
\newblock {\em \href{http://semr.math.nsc.ru}{Sib. \`Elektron. Mat. Izv.}},
  16:501--515, 2019.
\newblock \DOI{10.33048/semi.2019.16.032}.

\bibitem{Wanless:survey2011}
I.~M. Wanless.
\newblock Transversals in {L}atin squares: a survey.
\newblock In {\em Surveys in Combinatorics 2011}, volume 392 of {\em London
  Math. Soc. Lecture Note Ser.}, pages 403--437. Cambridge Univ. Press,
  Cambridge, 2011.

\end{thebibliography}

\providecommand\href[2]{#2} \providecommand\url[1]{\href{#1}{#1}}
  \def\DOI#1{{\small {DOI}:
  \href{http://dx.doi.org/#1}{#1}}}\def\DOIURL#1#2{{\small{DOI}:
  \href{http://dx.doi.org/#2}{#1}}}

\end{document}